\documentclass[english]{article}
\usepackage[margin=2.5cm]{geometry}
\usepackage{hyperref}
\usepackage{amsthm,amsfonts,amssymb,amsmath}
\usepackage{graphicx}

\newtheorem{theorem}{Theorem}

\newtheorem{claim}[theorem]{Claim}

\theoremstyle{definition}

\numberwithin{figure}{section}
\allowdisplaybreaks

\usepackage{babel}
\usepackage[T1]{fontenc}
\usepackage[utf8]{inputenc}
\usepackage{paralist}
\usepackage{verbatim}

\newcommand{\eps}{\varepsilon}
\newcommand{\su}{\subseteq}
\newcommand{\sm}{\setminus}

\newcommand{\F}{\mathbb{F}}

\newcommand{\alphsymb}{\Sigma}

\usepackage{color}
\usepackage[dvipsnames]{xcolor}

\title{List-decodability with large radius for Reed--Solomon codes}
\author{Asaf Ferber\thanks{Department of Mathematics, University of California, Irvine.
Email: \href{mailto:asaff@uci.edu} {\nolinkurl{asaff@uci.edu}}.
Research supported in part by NSF Awards DMS-1954395 and DMS-1953799.}
\and 
Matthew Kwan\thanks{Department of Mathematics, Stanford University, Stanford, CA.
Email: \href{mattkwan@stanford.edu}{\nolinkurl{mattkwan@stanford.edu}}.
Research supported by NSF Award DMS-1953990.}
\and
Lisa Sauermann\thanks{School of Mathematics, Institute for Advanced Study, Princeton, NJ. Email: \href{lsauerma@mit.edu}{\nolinkurl{lsauerma@mit.edu}}. Research supported by NSF Grant CCF-1900460 and NSF Award DMS-2100157.}
}

\begin{document}

\maketitle

\begin{abstract}
List-decodability of Reed--Solomon codes has received a lot of attention, but the best-possible dependence between the parameters is still not well-understood. In this work, we focus on the case where the list-decoding radius is of the form $r=1-\varepsilon$ for $\varepsilon$ tending to zero. Our main result states that there exist Reed--Solomon codes with rate $\Omega(\varepsilon)$ which are $(1-\varepsilon, O(1/\varepsilon))$-list-decodable, meaning that any Hamming ball of radius $1-\varepsilon$ contains at most $O(1/\varepsilon)$ codewords. This trade-off between rate and list-decoding radius is best-possible for any code with list size less than exponential in the block length.

By achieving this trade-off between rate and list-decoding radius we improve a recent result of Guo, Li, Shangguan, Tamo, and Wootters, and resolve the main motivating question of their work. Moreover, while their result requires the field to be exponentially large in the block length, we only need the field size to be polynomially large (and in fact, almost-linear suffices). We deduce our main result from a more general theorem, in which we prove good list-decodability properties of random puncturings of any given code with very large distance.
\end{abstract}

\section{Introduction}

Reed--Solomon codes are a family of error-correcting codes that have been studied intensively in many different contexts since they were introduced in \cite{reed-solomon}. As the parameters of the code, consider a prime power $q$ and integers $1\leq k<n\leq q$. Then, for $n$ distinct \emph{evaluation points} $\alpha_1,\dots,\alpha_n\in \F_q$, the $[n,k]$-Reed--Solomon code with these evaluation points is defined to be the set of \emph{codewords}
\[\mathcal{C}^{(k)}_{\alpha_1,\dots,\alpha_n}:=\{(f(\alpha_1),\ldots,f(\alpha_n)) \mid f\in \F_q[x],\, \deg f<k\}\su \F_q^n.\]

One reason for the great interest in Reed--Solomon codes is that they behave optimally with respect to the classical \emph{unique decoding} problem, having an optimal trade-off between rate and distance.
For an alphabet $\alphsymb$ of size $|\alphsymb|=q$, and a code $\mathcal C\subseteq \alphsymb^n$, the \emph{rate} of $\mathcal C$ is defined to be $\log_q |\mathcal C|/n$, and the \emph{distance} of $\mathcal C$ is defined to be the minimum Hamming distance between a pair of distinct codewords $\gamma,\gamma'\in \mathcal C$ (recall that the Hamming distance between $\gamma$ and $\gamma'$ is the number of positions in which $\gamma$ and $\gamma'$ disagree). Every $[n,k]$-Reed--Solomon code has rate $k/n$ and distance $n-k+1$. By the Singleton bound \cite{singleton}, this is the highest possible rate for any code of this distance. In addition, due to their simple structure, Reed--Solomon codes allow for efficient algorithms\footnote{In this paper, we are not concerned with algorithmic questions, and only study the combinatorial properties of Reed--Solomon codes.}.

An important generalization of the unique decoding problem is the problem of \emph{list-decoding}, and properties of Reed--Solomon codes with respect to list-decodability are much less understood. Roughly speaking, while the unique encoding problem demands that the original codeword can be uniquely reconstructed from a noisy signal, for the list-decoding problem we are satisfied with a short list of candidate codewords for a noisy signal. List-decodability was first introduced by Elias and Wozencraft \cite{elias, wozencraft} in the 1950s, and has since been used in several different areas of theoretical computer science. Regarding list-decodability of Reed--Solomon codes specifically, there are applications in complexity theory and the theory of pseudorandomness \cite{cai-et-al, lund-potukuchi, sudan-et-al}. The problem of understanding the (combinatorial) list-decodability of Reed--Solomon codes has been raised by many researchers over the last two decades (see for example \cite[p.\ 111]{guruswami-thesis}, \cite[p.\ 120]{rudra-thesis}, and \cite[Problem 5.20]{vadhan-lecture-notes}), and there has been a lot of recent work investigating this problem \cite{guo-et-al, rudra-wootters, shangguan-tamo}. Still, a lot of questions remain open.

In order to formally define what it means for a code to be list-decodable, we need to introduce some more definitions. Given $r\in (0,1)$, an alphabet $\alphsymb$, and $\beta\in \alphsymb^n$, the \emph{Hamming ball} of (relative) radius $r$ centered at $\beta$ is defined as
\[B_{r}(\beta):=\{\gamma\in \alphsymb^n \mid \gamma[i]=\beta[i] \textrm{ for at least }(1-r)n \textrm{ positions } 1\leq i\leq n\}\]
(here, by $\gamma[i]$ we denote the symbol in the $i$-th position of $\gamma\in\alphsymb^n$). In other words, this Hamming ball consists of all points $\gamma\in \alphsymb^n$ that differ in most $rn$ coordinates from $\beta$.

A code $\mathcal{C}\su \alphsymb^n$ is called \emph{$(r,L)$-list-decodable} (for some radius $r\in (0,1)$ and some \emph{list size} $L\in \mathbb{N}$) if we have $\vert \mathcal{C}\cap B_{r}(\beta)\vert\leq L$ for all $\beta\in \alphsymb^n$. In other words, $\mathcal{C}$ is $(r,L)$-list-decodable if each Hamming ball of (relative) radius $r$ in $\alphsymb^n$ contains at most $L$ codewords from $\mathcal{C}$. Note that for list size $L=1$, the setting of $(r,L)$-list-decodability precisely corresponds to the classical unique decoding setting. In this paper, we are primarily interested in list-decodability for Reed--Solomon codes.

For any radius $r\in (0,1)$ and any list size $L$, one can ask for the maximum possible rate of an $(r,L)$-list-decodable Reed--Solomon code. Shangguan and Tamo~\cite{shangguan-tamo} posed a precise conjecture for this general question, and made some partial progress towards their conjecture. Here, we focus on the case of radius $r=1-\eps$, for $\eps$ tending to zero. The main problem we investigate is how large the rate can be for a $(1-\eps,L)$-list-decodable $[n,k]$-Reed--Solomon code (for growing $n$), when the list size $L$ is not too large (say, not exponential in $n$).

There are some general results which immediately give upper and lower bounds for this problem. First, the \emph{list-decoding capacity theorem} (see for example \cite[Theorem 7.4.1]{guruswami-rudra-sudan-book}) implies that the rate of any  $(1-\eps,L)$-list-decodable code (where the list size $L$ is less than exponential in the block length $n$) can be at most\footnote{More precisely, as $n$ grows, the rate of such a code cannot be bounded above $\eps$.} $O(\eps)$. Second, the \emph{Johnson bound} \cite{johnson} gives a general bound on the list-decodability of a code in terms of its distance, and implies that \emph{every} $[n,k]$-Reed--Solomon code with rate $k/n=\eps^2$ is $(1-\eps,qn^2)$-list-decodable. In particular, there exist $(1-\eps,L)$-list-decodable Reed--Solomon codes that have rate $\eps^2$, where the list size $L$ is polynomial in $n$. Thus, the highest possible rate for the above problem lies somewhere between $\Omega(\eps^2)$ and $O(\eps)$.


Recently, Guo, Li, Shangguan, Tamo, and Wootters~\cite{guo-et-al} made major progress on closing this gap, improving the lower bound $\eps^2$ obtained from the Johnson bound\footnote{Rudra and Wootters \cite{rudra-wootters}, in an earlier work, also proved lower bounds that improve on the Johnson bound in certain regimes of $q$ and $\eps$. See also the comment further below.}. They proved that over very large fields there exist $(1-\eps,O(1/\eps))$-list-decodable Reed--Solomon codes with rate $\Omega(\eps/\log (1/\eps))$, matching the list-decoding capacity upper bound up to a logarithmic factor. They stated that their ``motivating question is whether or not RS codes can be list-decoded up to radius $1-\eps$ with rates $\Omega(\eps)$'', and this question remained open.

Our main result resolves this question in the affirmative, closing the gap to the list-decoding capacity upper bound (up to constant factors). We prove that over sufficiently large fields there exist $(1-\eps,O(1/\eps))$-list-decodable Reed--Solomon codes with rate $\Omega(\eps)$. This means that, up to constant factors, Reed--Solomon codes achieve the highest possible rate among all $(1-\eps,L)$-list-decodable codes where the list size $L$ is smaller than exponential in the block length $n$. A more precise statement of our main result is as follows.

\begin{theorem}\label{thm-main-simple}
Fix a constant $c\geq 5$. Let $\eps\in (0,1)$, let $n\in \mathbb{N}$ be sufficiently large with respect to $\eps$ and $c$, and let $q$ be a prime power with $q\geq n^{c/(c-1)}$. Then there exist $(1-\eps,\lceil 3/\eps\rceil)$-list-decodable $[n,k]$-Reed--Solomon codes over $\F_q$ with rate at least $\eps/(3c)$.
\end{theorem}

As mentioned above, the rate $\eps/(3c)$ here is tight up to the constant factor $3c$. Furthermore, Theorem~\ref{thm-main-simple} also improves the above-mentioned result of \cite{guo-et-al} in terms of the required field size $q$: in \cite{guo-et-al}, the field size $q$ needs to be exponential in the block length $n$, whereas Theorem~\ref{thm-main-simple} only assumes the polynomial bound $q\geq n^{c/(c-1)}$. By choosing a large constant $c$, the exponent $c/(c-1)$ can be taken arbitrarily close to $1$. In this sense, we can take the field size to be almost-linear in the block length $n$.

Similarly to the approach in \cite{guo-et-al}, we actually show that one can obtain the desired Reed--Solomon codes in Theorem~\ref{thm-main-simple} via a random choice of the evaluation points $(\alpha_1,\dots,\alpha_n)$: For suitable parameters $L=O(1/\eps)$, $n$ and $k=\Omega(\eps n)$, and for sufficiently large $q$, we prove that for almost all choices of $(\alpha_1,\dots,\alpha_n)\in \F_q^n$ the $[n,k]$-Reed--Solomon code $\mathcal{C}^{(k)}_{\alpha_1,\dots,\alpha_n}$ is $(1-\eps,L)$-list-decodable (and has rate $k/n=\Omega(\eps)$).

\begin{theorem}\label{thm-main-list-dec-rate}
Fix a constant $c\geq 5$. Let $\eps\in (0,1)$, let $n\in \mathbb{N}$ be sufficiently large with respect to $\eps$ and $c$, and let $k=\lceil \eps n/(3c)\rceil$. Furthermore, let $q$ be a prime power with $q\geq n^{c/(c-1)}$. Then for a uniformly random choice of an $n$-tuple $(\alpha_1,\dots,\alpha_n)\in \F_q^n$ with distinct entries $\alpha_1,\dots,\alpha_n$, the Reed--Solomon code $\mathcal{C}^{(k)}_{\alpha_1,\dots,\alpha_n}$ with rate $k/n\geq \eps/(3c)$ is $(1-\eps,\lceil 3/\eps\rceil)$-list-decodable with probability at least $1-q^{-\eps n/(13c)}$.
\end{theorem}

Note that Theorem~\ref{thm-main-list-dec-rate} immediately implies Theorem~\ref{thm-main-simple}. Our proof approach for Theorem~\ref{thm-main-list-dec-rate} is inspired by the approach of Guo, Li, Shangguan, Tamo, and Wootters in \cite{guo-et-al}, which in turn builds on the ideas in earlier work of Shangguan and Tamo \cite{shangguan-tamo}. However, our proof is significantly simpler and much shorter.

In fact, we deduce Theorem~\ref{thm-main-list-dec-rate} from a much more general result about random puncturings of arbitrary codes with very large distance. A \emph{puncturing} of a code $\mathcal C\subseteq \alphsymb^m$ to a set $S\subseteq [m]$ is defined to be the code $\mathcal C_S\subseteq\alphsymb^S$ whose codewords are obtained by restricting all the codewords in $\mathcal C$ to only the positions in $S$. Formally, $\mathcal C_S=\{(\gamma[i])_{i\in S}\mid \gamma\in \mathcal C\}$. We will consider random puncturings of a given code $\mathcal{C}$ obtained by choosing a uniformly random subset $S\subseteq [m]$ of a given size $n$ (then the puncturing $\mathcal{C}_S$ has block length $n$).

We prove the following general result concerning list-decodability of random puncturings of a given code with large distance. Roughly speaking, this result states that for a code with block length $m$ and distance $m-h$ (for some small $h$), a random puncturing with block length $n$ is likely to be list-decodable with radius $1-O(h/n)$ and list size $O(n/h)$, provided that the alphabet is large enough and $n$ is not too big. In order to deduce Theorem \ref{thm-main-list-dec-rate}, we apply Theorem~\ref{thm-general} to the ``full'' $[q,k]$-Reed--Solomon code $\mathcal{C}\su \F_q^q$ where the evaluation points are all of the $q$ points in $\F_q$.

\begin{theorem}
\label{thm-general}
Fix a constant $c\geq 5$, and let $q\in \mathbb{N}$ be sufficiently large with respect to $c$. Suppose that $h,m\in \mathbb{N}$ are such that $h\leq q^{-1/c}\cdot m$, and let $\mathcal{C}\su \alphsymb^m$ be a code over an alphabet $\alphsymb$ of size $\vert \alphsymb\vert =q$ such that $\mathcal{C}$ has distance at least $m-h$. Then for any $n\in \mathbb{N}$ satisfying
\[3c\cdot h< n\leq \min\left(\sqrt{\log_2 q}\cdot \sqrt{c/8}\cdot h\, ,\,  e^{h/(4c^3)}\cdot (c/2)\cdot h\right),\]
a random puncturing of $\mathcal{C}$ of block length $n$ is $(1-(3ch/n), \lfloor n/(ch)\rfloor)$-list-decodable with probability at least $1-q^{-h/4}$. In particular, there exist $(1-(3ch/n), \lfloor n/(ch)\rfloor)$-list-decodable puncturings of $\mathcal{C}$ of block length $n$.
\end{theorem}

We made no particular effort to optimize the constants in  the theorems above.

We remark that Rudra and Wootters~\cite{rudra-wootters} previously also proved results concerning list-decodability of random puncturings of codes with large distance. However, the details of their results and our Theorem \ref{thm-general} differ significantly. In particular, while our theorem requires a much larger distance of the code $\mathcal{C}$, in their results the block length $n$ of the puncturing needs to be larger. For this reason, with their results one cannot obtain Reed--Solomon codes with rates as large as in Theorem \ref{thm-main-list-dec-rate}.

Let us also briefly comment on some other works related to our main result, Theorem \ref{thm-main-list-dec-rate} (a more detailed review of the relevant literature can be found in \cite[Section 1.2]{guo-et-al}). Using their random puncturing results mentioned above, Rudra and Wootters \cite{rudra-wootters} proved a result similar to Theorem \ref{thm-main-list-dec-rate}, but only with rate $\Omega(\eps/(\log^5(1/\eps)\log q))$. This is weaker than our Theorem \ref{thm-main-list-dec-rate} and than the result of Guo, Li, Shangguan, Tamo, and Wootters \cite{guo-et-al}, and in particular due to the factor $\log q$ in the denominator the rate bound in \cite{rudra-wootters} always goes to zero as $n$ grows (since $q\geq n$). Shangguan and Tamo \cite{shangguan-tamo} proved a result of a similar spirit as Theorem~\ref{thm-main-list-dec-rate} for small list sizes $L=2$ and $L=3$ (which in particular means that the radius $r$ is bounded away from $1$), but with an optimal trade-off between radius, rate and list size (more precisely, for given rate and list size $L\in \{2,3\}$ their result gives the exact best-possible list-decoding radius). On a different note, while Theorem \ref{thm-main-list-dec-rate} shows that almost all choices of the evaluation points $(\alpha_1,\dots,\alpha_n)\in \F_q^n$ lead to $(1-\eps,\lfloor 10/\eps\rfloor)$-list-decodable Reed--Solomon codes, it is plausible that some choices of $(\alpha_1,\dots,\alpha_n)$ fail to have this property. There are some related results of Guruswami and Rudra \cite{guruswami-rudra} and of Ben-Sasson, Kopparty, and Radhakrishnan \cite{ben-sasson-et-al} pointing in this direction (\cite{guruswami-rudra} shows that for some choices of $(\alpha_1,\dots,\alpha_n)$ the code $\mathcal{C}^{(k)}_{\alpha_1,\dots,\alpha_n}$ fails to satisfy a stronger property called list-recoverability, and \cite{ben-sasson-et-al} shows a negative result concerning the list-decodability of Reed--Solomon codes in the case $q=n$ where up to permutation there is only one choice of the evaluation points $(\alpha_1,\dots,\alpha_n)$).

\textit{Notation.} Let $\mathbb{N}=\{1,2,3,\dots\}$, and for $n\in \mathbb{N}$ let $[n]=\{1,\dots,n\}$.

\section{Proofs}

Recall that Theorem~\ref{thm-main-list-dec-rate} implies Theorem~\ref{thm-main-simple}. We now show how Theorem~\ref{thm-main-list-dec-rate} follows from Theorem~\ref{thm-general}.

\begin{proof}[Proof of Theorem \ref{thm-main-list-dec-rate}]Let $c\geq 5$, let $\eps\in (0,1)$, and let $n\in \mathbb{N}$ be sufficiently large with respect to $\eps$ and $c$. Let $k=\lceil \eps n/(3c)\rceil$ and let $q$ be a prime power with $q\geq n^{c/(c-1)}$.

In order to apply Theorem \ref{thm-general}, let $m=q$ and $h=k-1\leq \eps n/(3c)$. Note that we have $h\leq n\leq q^{(c-1)/c}=q^{-1/c}\cdot m$ and
\[3c\cdot \lceil \eps n/(3c)\rceil < n\leq \min\left(\sqrt{\log_2 q}\cdot \sqrt{c/8}\cdot (\lceil \eps n/(3c)\rceil-1)\, ,\,  e^{(\lceil \eps n/(3c)\rceil-1)/(4c^3)}\cdot (c/2)\cdot (\lceil \eps n/(3c)\rceil-1)\right),\]
by the assumption that $n$ (and therefore also $q$) is sufficiently large with respect to $\eps$ and $c$.

Let us consider the alphabet $\alphsymb=\F_q$ and the ``full'' $[q,k]$-Reed--Solomon code $\mathcal{C}\su \F_q^q$ where the evaluation points are all of the $q$ points in $\F_q$. Note that $\mathcal{C}$ has distance $q-k+1= m-h$.

Hence all assumptions of Theorem \ref{thm-general} are satisfied, and we can conclude that a random puncturing of $\mathcal{C}$ of block length $n$ is $(1-(3ch/n), \lfloor n/(ch)\rfloor)$-list-decodable with probability at least $1-q^{-h/4}\geq 1-q^{-(\eps n/(12c))+(1/4)}\geq 1-q^{-\eps n/(13c)}$ (using that $n$ is sufficiently large with respect to $\eps$ and $c$). Noting that $1-(3ch/n)\geq 1-\eps$ and $n/(ch)\leq n/(\eps n/3-c)<(3/\eps)+1$ (again, as $n$ is sufficiently large with respect to $\eps$ and $c$), this implies that such a random puncturing of $\mathcal{C}$ is $(1-\eps, \lceil 3/\eps\rceil)$-list-decodable with probability at least $1-q^{-\eps n/(13c)}$. In other words, for a uniformly random choice of an $n$-tuple $(\alpha_1,\dots,\alpha_n)\in \F_q^n$ with distinct entries $\alpha_1,\dots,\alpha_n$, the Reed--Solomon code $\mathcal{C}^{(k)}_{\alpha_1,\dots,\alpha_n}$ is $(1-\eps,\lceil 3/\eps\rceil)$-list-decodable with probability at least $1-q^{-\eps n/(13c)}$.
\end{proof}

Our aim for the rest of this section is to prove Theorem \ref{thm-general}. We deduce Theorem \ref{thm-general} from the following theorem. This approach is motivated by \cite{guo-et-al} and \cite{shangguan-tamo}, even though the setting there is specific to Reed--Solomon codes.

\begin{theorem}\label{thm-weaker-list-decoding-4kL}
Fix a constant $c\geq 5$, suppose that $q,h,m\in \mathbb{N}$ are such that $h\leq q^{-1/c}\cdot m$, and let $\mathcal{C}\su \alphsymb^m$ be a code over an alphabet $\alphsymb$ of size $\vert \alphsymb\vert =q$ such that $\mathcal{C}$ has distance at least $m-h$.

Suppose $L$ is a non-negative integer satisfying $L<e^{h/(4c^3)}-2$. Let $n\in \mathbb{N}$ and consider subsets $I_1,\dots,I_{L+1}\su [n]$ such that
\begin{equation}\label{eq-weight-sets-I}
\sum_{j=1}^{L+1}\vert I_j\vert-\left\vert\bigcup_{j=1}^{L+1} I_j\right\vert> 2chL.
\end{equation}
Let us say that an $n$-tuple $(a_1,\dots,a_n)\in [m]^n$ with distinct entries $a_1,\dots,a_n$ is \emph{bad} if there exist a point $\beta\in \alphsymb^m$ and distinct codewords $\gamma_1,\dots,\gamma_{L+1}\in \mathcal{C}$ such that for all $j=1,\dots,L+1$ and all $i\in I_j$ we have $\gamma_j[a_i]=\beta[a_i]$. Then there are at most $q^{-h/2}\cdot m^n$ bad $n$-tuples $(a_1,\dots,a_n)\in [m]^n$.
\end{theorem}

Guo, Li, Shangguan, Tamo, and Wootters \cite[Theorem 6.3]{guo-et-al} proved a statement similar to Theorem \ref{thm-weaker-list-decoding-4kL} in the specific setting of Reed--Solomon codes. However, their statement gives a weaker bound for the number of bad $n$-tuples and requires a stronger version of the assumption (\ref{eq-weight-sets-I}), where the term on the right-hand side of (\ref{eq-weight-sets-I}) is larger by a factor of $\Theta(\log L)$. This additional logarithmic factor leads to the logarithmic loss in the rate $\Omega(\eps/\log (1/\eps))$ of the  Reed--Solomon codes in their result (and the weaker bound for the number of bad $n$-tuples leads to them requiring the field size $q$ to be exponential in the block length $n$).

Let us now show the deduction of Theorem \ref{thm-general} from Theorem \ref{thm-weaker-list-decoding-4kL}. This deduction is fairly standard (similar arguments appear in \cite{guo-et-al,shangguan-tamo}). Afterwards, at the end of this section, we will present the proof of Theorem \ref{thm-weaker-list-decoding-4kL}.

\begin{proof}[Proof of Theorem \ref{thm-general}]
Let us define $L=\lfloor n/(ch)\rfloor\geq 3$, and note that by the assumptions on $n$ we have
\[L+2\leq 2L\leq \frac{2n}{ch}\leq e^{h/(4c^3)}.\]
Also note that
\begin{equation}\label{eq-ineq-paramaters}
\frac{n+2chL}{L+1}< \frac{n}{L+1}+2ch\leq \frac{n}{n/(ch)}+2ch=3ch.
\end{equation}
Let us also remark that the assumptions in Theorem \ref{thm-general} (including the assumption that $q$ is sufficiently large with respect to $c$) imply that
\begin{equation}\label{eq-ineq-paramaters-2}
\frac{n}{m}\leq \frac{\sqrt{\log_2 q}\cdot \sqrt{c/8}\cdot h}{q^{1/c}\cdot h}=\sqrt{c/8}\cdot \frac{\sqrt{\log_2 q}}{q^{1/c}}< \frac{1}{2}
\end{equation}
and
\begin{equation}\label{eq-ineq-paramaters-3}
\frac{2n^2}{m}\leq \frac{2\cdot \log_2 q\cdot (c/8)\cdot h^2}{q^{1/c}\cdot h}=\frac{1}{4}\cdot c\cdot h\cdot \log_2 q\cdot q^{-1/c}<\frac{1}{12}\cdot h\cdot \log_2 q.
\end{equation}

We need to show that a (uniformly) random puncturing of $\mathcal{C}$ of block length $n$ is $(1-(3ch/n), L)$-list-decodable with probability at least $1-q^{-h/4}$. We can model the choice of such a random puncturing by taking a uniformly random $n$-tuple $(a_1,\dots,a_n)\in [m]^n$ with distinct entries $a_1,\dots,a_n$ and considering the puncturing $\mathcal{C}_S$ for $S=\{a_1,\dots,a_n\}$. Note that $\mathcal{C}_S$ fails to be $(1-(3ch/n), L)$-list-decodable if and only if there exist a point $\beta\in \alphsymb^m$ and distinct codewords $\gamma_1,\dots,\gamma_{L+1}\in \mathcal{C}$ such that for each $j=1,\dots,L+1$ we have $\gamma_j[s]=\beta[s]$ for at least $3ch$ elements $s\in S$. Recalling that $S=\{a_1,\dots,a_n\}$, this condition is equivalent to having $\gamma_j[a_i]=\beta[a_i]$ for at least $3ch$ indices $i\in [n]$.

Hence, if for our random choice of $(a_1,\dots,a_n)\in [m]^n$ the puncturing $\mathcal{C}_S$ fails to be $(1-(3ch/n), L)$-list-decodable, then for each $j=1,\dots,L+1$ we can find a set $I_j\su [n]$ of size $\vert I_j\vert\geq 3ch$ such that we have $\gamma_j[a_i]=\beta[a_i]$ for all $i\in I_j$. With the notation in Theorem \ref{thm-weaker-list-decoding-4kL}, this means that the $n$-tuple $(a_1,\dots,a_n)\in [m]^n$ is bad with respect to the subsets $I_1,\dots,I_{L+1}$.

Note that there are at most $(2^{n})^{L+1}$ possibilities to choose subsets $I_1,\dots,I_{L+1}\su [n]$ with $\vert I_j\vert\geq 3ch$ for $j=1,\dots,L+1$. For any such choice of subsets, by (\ref{eq-ineq-paramaters}) we have
\[\sum_{j=1}^{L+1}\vert I_j\vert-\left\vert\bigcup_{j=1}^{L+1} I_j\right\vert\geq (L+1)\cdot 3ch-n> (L+1)\cdot \frac{n+2chL}{L+1}-n=2chL.\]
Hence, by Theorem \ref{thm-weaker-list-decoding-4kL}, for any fixed choice of $I_1,\dots,I_{L+1}$, there are at most $q^{-h/2}\cdot m^n$ different $n$-tuples $(a_1,\dots,a_n)\in [m]^n$ which are bad with respect to $I_1,\dots,I_{L+1}$. Overall, this means that there are at most $2^{n(L+1)}\cdot q^{-h/2}\cdot m^n$ different $n$-tuples $(a_1,\dots,a_n)\in [m]^n$ which are bad with respect to some choice of subsets $I_1,\dots,I_{L+1}\su [n]$ (with $\vert I_j\vert\geq 3ch$ for $j=1,\dots,L+1$). Thus, the number of $n$-tuples $(a_1,\dots,a_n)\in [m]^n$ with distinct entries $a_1,\dots,a_n$, such that the puncturing $\mathcal{C}_S$ for $S=\{a_1,\dots,a_n\}$ is not $(1-(3ch/n), L)$-list-decodable, is at most
\[2^{n(L+1)}\cdot q^{-h/2}\cdot m^n\leq 2^{(4/3)Ln}\cdot q^{-h/2}\cdot m^n \leq 2^{(4/3)n^2/(ch)}\cdot q^{-h/2}\cdot m^n\leq q^{h/6}\cdot q^{-h/2}\cdot m^n=q^{-h/3}\cdot m^n,\]
where for the third inequality we used the assumption that $n\leq \sqrt{\log_2 q}\cdot \sqrt{c/8}\cdot h$.

Finally, note that the total number of $n$-tuples $(a_1,\dots,a_n)\in [m]^n$ with distinct entries $a_1,\dots,a_n$ is
\[m(m-1)\dotsm(m-n+1)\geq \left(1-\frac{n}{m}\right)^n\cdot m^n\geq 2^{-2n^2/m}\cdot m^n\geq q^{-h/12}\cdot m^n.\]
Here, we used that $1-x\geq 2^{-2x}$ for all $x\in (0,1/2)$, as well as (\ref{eq-ineq-paramaters-2}) and (\ref{eq-ineq-paramaters-3}).

All in all, this means that for a random choice of an $n$-tuple $(a_1,\dots,a_n)\in [m]^n$ with distinct entries $a_1,\dots,a_n$, the probability that the puncturing $\mathcal{C}_S$ for $S=\{a_1,\dots,a_n\}$ fails to be $(1-(3ch/n), L)$-list-decodable is at most
\[\frac{q^{-h/3}\cdot m^n}{q^{-h/12}\cdot m^n}=q^{-h/4}.\]
Hence a random puncturing of $\mathcal{C}$ of block length $n$ is $(1-(3ch/n), L)$-list-decodable with probability at least $1-q^{-h/4}$, as desired.
\end{proof}

It remains to prove Theorem \ref{thm-weaker-list-decoding-4kL}. This is the part of this paper requiring new ideas. Roughly speaking, the proof strategy is as follows. Recall that an $n$-tuple $(a_1,\dots,a_n)\in [m]^n$ is called bad if there are distinct codewords $\gamma_1,\dots,\gamma_{L+1}\in \mathcal C$ and a point $\beta\in \alphsymb^m$ such that $\gamma_j[a_i]=\beta[a_i]$ whenever $i\in I_j$. Our goal is to prove an upper bound on the number of bad $n$-tuples $(a_1,\dots,a_n)$. The key idea of the proof is to find a relatively small set of indices $Z\su [n]$, such that specifying $a_i$ and $\beta[a_i]$ for all $i\in Z$ already uniquely determines all of the codewords $\gamma_1,\dots,\gamma_{L+1}$ (via the condition that $\gamma_j[a_i]=\beta[a_i]$ whenever $i\in I_j$, and the assumption that $\mathcal{C}$ has large distance). Once the codewords $\gamma_1,\dots,\gamma_{L+1}$ are determined, for any distinct $j,j'\in \{1,\dots,L+1\}$ and any $i\in I_j\cap I_{j'}$, there are only a small number of choices for $a_i\in [m]$. Indeed, we must have $\gamma_j[a_i]=\beta[a_i]=\gamma_{j'}[a_i]$, so $a_i$ must be one of the few positions in which the codewords $\gamma_j$ and $\gamma_{j'}$ agree. Overall, we obtain the desired upper bound for the number of bad $n$-tuples $(a_1,\dots,a_n)$ by a counting argument that takes all of these restricted choices into account.

\begin{proof}[Proof of Theorem \ref{thm-weaker-list-decoding-4kL}]
We prove the theorem by induction on $L$. First, note that the statement is vacuously true for $L=0$, because it is impossible for the condition $|I_1|-|I_1|>2ch\cdot 0$ in (\ref{eq-weight-sets-I}) to be satisfied.

Let us now assume that $L\geq 1$, and that we have already proved the theorem for $L-1$. First, we consider the case that for some index $t\in \{1,\dots,L+1\}$ we have
\[\left\vert I_t\,\cap\bigcup_{j\in \{1,\dots,L+1\}\sm   \{t\}}I_j\right\vert<2ch.\]
Let us assume without loss of generality that $t=L+1$, then we have
\[\left\vert I_{L+1}\,\cap\bigcup_{j=1}^{L}I_j\right\vert<2ch.\]
But now (\ref{eq-weight-sets-I}) implies that
\[\sum_{j=1}^{L}\vert I_j\vert-\left\vert\bigcup_{j=1}^{L} I_j\right\vert=\sum_{j=1}^{L+1}\vert I_j\vert-\left\vert\bigcup_{j=1}^{L+1} I_j\right\vert-\left\vert I_{L+1}\,\cap\bigcup_{j=1}^{L}I_j\right\vert> 2chL-2ch=2ch(L-1).\]
This means that we can apply the induction hypothesis to $L-1$ and the sets $I_1,\dots,I_L$. This shows that the number of bad $n$-tuples $(a_1,\dots,a_n)\in [m]^n$ is at most $q^{-h/2}\cdot m^n$, since every $n$-tuple which is bad for the sets $I_1,\dots,I_{L+1}$ must also be bad for the sets $I_1,\dots,I_{L}$.

So we may from now on assume that for all $t=1,\dots, L+1$ we have
\[\left\vert I_t\,\cap\bigcup_{j\in \{1,\dots,L+1\}\sm   \{t\}}I_j\right\vert\geq 2ch.\]

Now let $M\su [n]$ be the set of those elements $i\in [n]$ that are contained in at least two of the sets $I_1,\dots,I_{L+1}$. Note that for each $t=1,\dots,L+1$, we have
\[\vert M\cap I_t\vert=\left\vert I_t\,\cap\bigcup_{j\in \{1,\dots,L+1\}\sm   \{t\}}I_j\right\vert\geq 2ch.\]
In particular, we have $\vert M\vert\geq 2ch$.

\begin{claim}\label{claim-set-Z}
There exists a set $Z\su M$ of size $\vert Z\vert\leq \vert M\vert/(2c-2)$ such that $\vert Z\cap I_t\vert> h$ for all $t=1,\dots,L+1$.
\end{claim}

\begin{proof}Let us choose the set $Z\su M$ randomly by including each element of $M$ into the set $Z$ independently with probability $1/(2c-1)$. By the Chernoff bound (see for example \cite[Theorem A.1.4]{alon-spencer}), we have that
\[\Pr\left(\vert Z\vert> \frac{\vert M\vert}{2c-2}\right)< \exp\left(-\frac{2}{\vert M\vert}\cdot \left(\frac{\vert M\vert}{(2c-2)(2c-1)}\right)^2\right)\leq e^{-\vert M\vert/(8c^4)}\leq e^{-h/(4c^3)}.\]
Furthermore, for each $t=1,\dots,L+1$, each element of $M\cap I_t$ is an element of the set $Z$ independently with probability $1/(2c-1)$. Hence, again by the Chernoff bound, we have (recalling that $\vert M\cap I_t\vert\geq 2ch$)
\begin{multline*}
    \Pr(\vert Z\cap I_t\vert\leq h)\leq \exp\left(-\frac{2}{\vert M\cap I_t\vert}\cdot \left(\frac{\vert M\cap I_t\vert}{2c-1}-h\right)^2\right)=\exp\left(-2\vert M\cap I_t\vert\cdot \left(\frac{1}{2c-1}-\frac{h}{\vert M\cap I_t\vert}\right)^2\right)\\
    \leq \exp\left(-2\cdot 2ch\cdot \left(\frac{1}{2c-1}-\frac{1}{2c}\right)^2\right)=\exp\left(-\frac{4ch}{(2c)^2(2c-1)^2}\right)\leq e^{-h/(4c^3)}.
\end{multline*}
All in all, by a union bound, the probability of having $\vert Z\vert\leq \vert M\vert/(2c-2)$ and  $\vert Z\cap I_t\vert> h$ for all $t=1,\dots,L+1$ is at least
\[1-e^{-h/(4c^3)}-(L+1)\cdot e^{-h/(4c^3)}=1-(L+2)\cdot e^{-h/(4c^3)}>0\]
(recalling our assumption that $L<e^{h/(4c^3)}-2$). This means that the desired set $Z\su M$ exists.
\end{proof}

Let us now fix a set $Z\su M$ as in Claim \ref{claim-set-Z}. Now we can show the desired upper bound on the number of bad $n$-tuples $(a_1,\dots,a_n)\in [m]^n$ in the following way. Recall that for a bad $n$-tuple $(a_1,\dots,a_n)\in [m]^n$ there exist a point $\beta\in \alphsymb^m$ and distinct codewords $\gamma_1,\dots,\gamma_{L+1}\in \mathcal{C}$ such that for all $j=1,\dots,L+1$ and all $i\in I_j$ we have $\gamma_j[a_i]=\beta[a_i]$.

Note that we have at most $m^{\vert Z\vert}$ choices for the elements $a_i$ for all $i\in Z$ (recall that the elements $a_i$ need to all be distinct). Furthermore, there are $\vert \alphsymb\vert^{\vert Z\vert}=q^{\vert Z\vert}$ possibilities for the values $\beta[a_i]$ for all $i\in Z$. Now, knowing $a_i$ and $\beta[a_i]$ for all $i\in Z$ already determines the codewords $\gamma_1,\dots,\gamma_{L+1}$. Indeed, for each $t=1,\dots,L+1$ we have $\vert Z\cap I_t\vert> h$ and $\gamma_t[a_i]=\beta[a_i]$ for all $i\in Z\cap I_t$ (and the coordinates $a_i\in [m]$ for $i\in Z\cap I_t$ are distinct). Since any two codewords in $\mathcal{C}$ agree in at most $h$ positions (as $\mathcal{C}$ has distance at least $m-h$), for each $t=1,\dots,L+1$ there is at most one possible codeword $\gamma_t$ satisfying $\gamma_t[a_i]=\beta[a_i]$ for all $i\in Z\cap I_t$. Thus, after choosing $a_i$ and $\beta[a_i]$ for all $i\in Z$, there is at most one possibility for the codewords $\gamma_1,\dots,\gamma_{L+1}$.

Furthermore, knowing the codewords $\gamma_1,\dots,\gamma_{L+1}$ there are at most $h$ possibilities for each $a_i\in [m]$ with $i\in M\sm Z$. Indeed, for each $i\in M\sm Z$ there exist two distinct indices $j,j'\in \{1,\dots,L+1\}$ with $i\in I_j\cap I_{j'}$ and we must have $\gamma_j[a_i]=\beta[a_i]=\gamma_{j'}(a_i)$. Hence the codewords $\gamma_j$ and $\gamma_{j'}$ must agree in position $a_i$. However, as $\mathcal{C}$ has distance at least $m-h$, the codewords $\gamma_j$ and $\gamma_{j'}$ agree in at most $h$ positions, and so there are at most $h$ possible choices for $a_i$. Thus, for each $i\in M\sm Z$, there are indeed at most $h$ choices for $a_i$ and altogether this gives at most $h^{\vert M\vert-\vert Z\vert}$ choices for determining all the the elements $a_i\in [m]$ with $i\in M\sm Z$.

Finally, there are at most $m^{n-\vert M\vert}$ choices for the elements $a_i\in [m]$ with $i\in [n]\sm M$. All in all, this means that the number of possible choices for a bad $n$-tuple $(a_1,\dots,a_n)\in [m]^n$ is at most
\begin{multline*}
m^{\vert Z\vert}\cdot q^{\vert Z\vert}\cdot h^{\vert M\vert-\vert Z\vert}\cdot m^{n-\vert M\vert}=\left(\frac{h}{m}\right)^{\vert M\vert} \left(\frac{qm}{h}\right)^{\vert Z\vert} m^n\\
\leq \left(\frac{h}{m}\right)^{\vert M\vert} \left(\frac{qm}{h}\right)^{\vert M\vert/(2c-2)} m^n=\left(\frac{h}{m}\cdot q^{1/(2c-3)}\right)^{\frac{2c-3}{2c-2}\cdot \vert M\vert} m^n\\
\leq \left(q^{-1/c}\cdot q^{1/(2c-3)}\right)^{\frac{2c-3}{2c-2}\cdot \vert M\vert} m^n=q^{-\frac{c-3}{2c-2}\cdot \frac{1}{c}\cdot \vert M\vert}m^n\leq q^{-\vert M\vert/(4c)}m^n\leq q^{-h/2}m^n.
\end{multline*}
Here, we used the assumptions $h\leq q^{-1/c}\cdot m$ and $c\geq 5$ as well as $\vert Z\vert\leq \vert M\vert/(2c-2)$ and $\vert M\vert\geq 2ch$.
\end{proof}

\section{Concluding remarks}
We have proved that there exist Reed--Solomon codes which are list-decodable with radius $1-\varepsilon$ (and polynomial list size) and have rate $\Omega(\varepsilon)$. Moreover, such codes exist with block length $n$ and field size $q$ whenever $n$ is sufficiently large and $q\ge n^{1+\delta}$, for any constant $\delta>0$. There are several interesting further directions of research.

First, our result uses the probabilistic method and is fundamentally non-constructive. It would be very interesting if, in the setting of Theorem~\ref{thm-main-list-dec-rate}, one could achieve the same bound with \emph{explicit} choices of evaluation points $\alpha_1,\dots,\alpha_n\in \F_q$. In fact, it would be interesting if one could beat the Johnson bound at all with an explicit Reed--Solomon code (we remark that there are constructions in \cite{guo-et-al,shangguan-tamo} which are in a certain sense explicit, but they require an exponential field size and therefore do not lead to efficient algorithms).

Second, it would be interesting to further improve the bounds in Theorem~\ref{thm-main-simple}. While our field size requirement $q\ge n^{1+\delta}$ is much weaker than the requirement in \cite{guo-et-al}, it would still be interesting to sharpen this further: does it suffice to assume that $q\ge Cn$ for some constant $C$? Also, it would be nice to optimize the constant factors in the trade-off between the rate and the list-decoding radius. In particular, it seems likely that there should exist Reed--Solomon codes which are list-decodable with radius $1-\varepsilon$ (and polynomial list size) and have rate $(1-o(1))\varepsilon$. An exact conjecture for the best-possible relationship between rate, list-decoding radius and list size was made by Shangguan and Tamo~\cite{shangguan-tamo}.

\textit{Acknowledgements.} We would like to thank Shachar Lovett for introducing us to list-decodability of Reed--Solomon codes, and Avi Wigderson for many very helpful suggestions.

\end{document}